\pdfoutput=1
\documentclass[acmsmall,screen, natbib=false]{acmart}
\setcopyright{none}
\copyrightyear{2024}
\acmYear{2024}
\acmDOI{}

\usepackage{fancyhdr}

\PassOptionsToPackage{strict}{csquotes} %
\usepackage{mathtools} %
\usepackage{booktabs}
\usepackage{hyperref}
\usepackage{tikz}
\usepackage{pgfplots}
\usepackage[capitalize,nameinlink]{cleveref} %
\RequirePackage[backend=biber,%
            sortcites=true,%
            maxcitenames=2,%
            minbibnames=3,
            maxbibnames=5, %
            natbib=true,%
                style=ACM-Reference-Format,
            language=american]{biblatex}
\usepackage{nag} %
\usepackage{fvextra}
\usepackage[htt]{hyphenat}
\usepackage{csquotes} %
\DeclareNameFormat{labelname:poss}{%
  \nameparts{#1}%
  \ifcase\value{uniquename}%
    \usebibmacro{name:family}{\namepartfamily}{\namepartgiven}{\namepartprefix}{\namepartsuffix}%
  \or
    \ifuseprefix{\usebibmacro{name:first-last}{\namepartfamily}{\namepartgiveni}{\namepartprefix}{\namepartsuffixi}}
      {\usebibmacro{name:first-last}{\namepartfamily}{\namepartgiveni}{\namepartprefixi}{\namepartsuffixi}}%
  \or
    \usebibmacro{name:first-last}{\namepartfamily}{\namepartgiven}{\namepartprefix}{\namepartsuffix}%
  \fi
  \usebibmacro{name:andothers}%
  \ifnumequal{\value{listcount}}{\value{liststop}}{'s}{}}
\DeclareFieldFormat{shorthand:poss}{%
  \ifnameundef{labelname}{#1's}{#1}}
\DeclareFieldFormat{citetitle:poss}{\mkbibemph{#1}'s}
\DeclareFieldFormat{label:poss}{#1's}
\newrobustcmd*{\posscitealias}{%
  \AtNextCite{%
    \DeclareNameAlias{labelname}{labelname:poss}%
    \DeclareFieldAlias{shorthand}{shorthand:poss}%
    \DeclareFieldAlias{citetitle}{citetitle:poss}%
    \DeclareFieldAlias{label}{label:poss}}}
\newrobustcmd*{\citeposs}{%
  \posscitealias%
  \textcite}
\newrobustcmd*{\Citeposs}{\bibsentence\citeposs}
\newrobustcmd*{\Citeposses}{%
  \posscitealias%
  \textcites}

\newcommand{\floor}[1]{\lfloor #1 \rfloor}

\newcommand{\figsize}{\small}

\def\truesymbol{{\tt \#}{\sf t}}
\def\falsesymbol{{\tt \#}{\sf f}}

\def\truth#1{\def\foo{#1}\def\goo{t}%
               \ifx\foo\goo{\truesymbol}\else{\falsesymbol}\fi}

\pgfplotsset{compat=1.18}

\begin{document}

\title{Six Ways to Implement Divisibility by Three in miniKanren}

\author{Brett Schreiber} \email{schreiberbrett@gmail.com} \affiliation{%
  \institution{Independent Researcher}
  \city{Philadelphia}
  \state{PA}
  \country{USA} }

\author{Brysen Pfingsten} \email{pfingsbr@shu.edu} \affiliation{%
  \institution{Seton Hall University}
  \city{South Orange}
  \state{NJ}
  \country{USA} }

\author{Jason Hemann} \email{hemannja@shu.edu} \affiliation{%
  \orcid{0000-0002-5405-2936}
  \institution{Seton Hall University}
  \city{South Orange}
  \state{NJ}
  \country{USA} }

\settopmatter{printacmref=false}
\settopmatter{printfolios=true}
\renewcommand\footnotetextcopyrightpermission[1]{}
\pagestyle{fancy}
\fancyfoot{}
\fancyfoot[R]{miniKanren'24}
\fancypagestyle{firstfancy}{
  \fancyhead{}
  \fancyhead[R]{miniKanren'24}
  \fancyfoot{}
}
\makeatletter
\let\@authorsaddresses\@empty
\makeatother

\begin{abstract}
  This paper explores options for implementing the relation $n \equiv 0 \ (\text{mod} \ 3)$ within miniKanren using miniKanren numbers and its arithmetic suite.
  We examine different approaches starting from straightforward implementations to more optimized versions.
  The implementations discussed include brute-force arithmetic methods, divisibility tricks, and derivation from a finite automaton.
  Our contributions include an in-depth look at the process of implementing a miniKanren relation and observations on benchmarking \texttt{defrel}s.
  This study aims to provide practical insights for miniKanren programmers on both performance and implementation techniques.
\end{abstract}

\keywords{miniKanren, logic programming, automata}

\ccsdesc[500]{Software and its engineering~Domain specific languages}
\ccsdesc[500]{Software and its engineering~Functional languages}
\ccsdesc[500]{Software and its engineering~Constraint and logic languages}
\ccsdesc[300]{Theory of computation~Formal languages and automata theory}
\ccsdesc[500]{Software and its engineering~Constraint and logic languages}

\maketitle{}
\thispagestyle{firstfancy}

\section{Introduction}\label{sec:introduction}

When a miniKanren programmer figures out which relation she wants to implement, that is just the beginning of the relational programming exercise.
Programmers in traditional languages know that many different procedures can implement the same mathematical function.
Likewise, the relational programmer is intimately familiar with the distinction between a mathematical relation and a \texttt{defrel} that implements it.
Constraint logic programming (CLP) offers the promise of clear, readable, unambiguous solutions that can be efficiently executed.
Every miniKanren programmer knows, however, that the \texttt{defrel}'s in the details.
Two \texttt{defrel}s implementing the same relation can differ asymptotically in their performance or in the way that they characterize the answer set.
This paper presents a case study illustrating where and how these issues can arise in practice.
Our working example is the unary relation describing the numbers evenly divisible by three, expressed in the usual list-based little-endian binary miniKanren numeral system (i.e.\ \enquote{Oleg numbers}).

We present six different one-argument predicates that implement this relation.
We start with succinct implementations that perform sub-optimally and we gently move towards more efficient implementations that unfold naturally from a corresponding function.
In the course of this exploration we shed light on the process of implementing a miniKanren relation and articulate heretofore unpublished folklore concerning difficulties benchmarking \texttt{defrel}s against each other.

\Cref{sec:minikanren-arithmetic} briefly reprises some important features of arithmetic in miniKanren.
In \Cref{sec:brute-force}, we present a selection of brute-force approaches to the problem.
In \Cref{sec:exploiting-the-structure-of-oleg-numbers} we exploit a divisibility trick for binary natural numbers and discover an improved implementation using partially instantiated miniKanren numbers to produce more general solutions.
We discuss the performance of our different implementations, and implications for benchmarking, in \Cref{sec:performance}.
Finally, we conclude by showing further opportunities for miniKanren to learn from the bit-tricks of hardware design, and some discussion of the broader implications of this exercise for miniKanren programmers generally.
\paragraph{Note to the reader} In the following \lcnamecrefs{sec:brute-force} we are careful to distinguish between the mathematical \textit{relation} and the \texttt{defrel} that (perhaps imperfectly) implements it.
We will, likewise, carefully distinguish between the number $n$ and its binary numeral expression \verb|n|.
Each relation implementation \texttt{impl} of \texttt{div3o} follows the naming convention \texttt{div3o/\textit{\{}impl\textit{\}}}.
Those implementations are:
\begin{itemize}
\item[-] \texttt{div3o/3*x}
\item[-] \texttt{div3o/x*3}
\item[-] \texttt{div3o/3+x}
\item[-] \texttt{div3o/x+3}
\item[-] \texttt{div3o/dfa}
\item[-] \texttt{div3o/even-odd}
\end{itemize}.
The code in this document uses the \texttt{faster-miniKanren}\footnote{\url{https://github.com/michaelballantyne/faster-minikanren}} implementation.

\section{Formalizing relational divisibility}\label{sec:minikanren-arithmetic}

We adopt the usual miniKanren pure relational arithmetic suite of \citet{friedman2018reasoned}.
Numbers are represented as little-endian binary lists, hereafter called \enquote{Oleg numerals} (or Oleg numbers, in a slight abuse of terminology).
\Citet{kiselyov2008pure} establish the theory of pure relational arithmetic and the key results.
We briefly reprise here some of their work and recall the necessary logic programming background (from, e.g.,\ \citet{lloyd1987foundations}, \citet{baader2001unification}).
In so doing we will introduce our problem of divisibility by three by way of a more familiar example: divisibility by \emph{two}.
Specific miniKanren examples in this \lcnamecref{sec:minikanren-arithmetic} are drawn from \citeauthor{friedman2018reasoned} or \citeauthor{kiselyov2008pure}.

Let terms $t_{1}$, $t_{2}$ be such that $t_{1}\theta = t_{2}$ for some substitution $\theta$.
In this case we say that $t_{1}$ is \emph{more specific than} $t_{2}$, and $t_{2}$ is \emph{more general than} $t_{1}$\footnote{ISO Prolog~\cite{ISO13211} provides the predicate \texttt{subsumes\_term/2} that directly implements these relationships.}.
The \emph{most specific generalization}~\cite{Flener1995} of two terms $t_{1}$ and $t_{2}$ is a term $t$ such that $t_{1} = t\theta_{1}$ and $t_{2} = t\theta_{2}$ for some substitutions $\theta_{1}$, $\theta_{2}$ and such that for any other term $s$ that generalizes $t_{1}$ and $t_{2}$, $s$ also generalizes $t$.

A \emph{goal} $g$ is understood as a function from a substitution to a finite or infinite stream of substitutions.
Goals can succeed or fail, and a goal can succeed multiple times.
For a miniKanren goal $g[x_{1},\ldots{},x_{n}]$, a \emph{solution} is an $n$-tuple of terms $(t_{1},\ldots{},t_{n})$ where each $t_{i}$ instantiates $x_{i}$ after the evaluation of $g[x_{1},\ldots{},x_{n}]$ at some point in the search.
We can equivalently represent a solution as $(x_{1},\ldots{},x_{n})\theta$, for some substitution $\theta$.

The miniKanren language presents solutions to the query goal with respect to the set of \texttt{run} variables provided with the query.
In miniKanren solutions are \emph{reified}, meaning a solution's uninstantiated variables are consistently replaced by canonical constants (written \enquote{\texttt{\_.}\textit{n}}), numbered by order of their first occurrence left to right in the solution.
Each logic variable is replaced by the same identifier across all its occurrences.
For instance, the query to \texttt{poso} in \cref{fig:poso} produces one solution, \verb|(_0 . _.1)|, because the query ensures that variable \texttt{q} has the structure of a pair.

\begin{figure}[ht]
\figsize{}
\begin{Verbatim}
(defrel (poso n)
  (fresh (a d)
    (== n `(,a . ,d))))

> (run 1 (q) (poso q))
'((_.0 . _.1))
\end{Verbatim}
\Description[]{}
\caption{The \texttt{poso} relation.}\label{fig:poso}
\end{figure}

A \emph{solution sequence} is a sequence of solutions obtained by repeatedly evaluating a goal in the execution of the language's backtracking search for solutions.
The precise search strategy matters; we of course use miniKanren's complete interleaving depth-first search~\cite{kiselyov2005backtracking}.
Some queries will have only finitely many solutions, other have infinitely many.
In contrast with the \texttt{poso} query, a \texttt{run} query for \texttt{(olego q)} using the \texttt{defrel} in \cref{fig:olego} can produce arbitrarily many solutions.

\begin{figure}[ht]
\figsize{}
\begin{Verbatim}
(defrel (olego n)
  (conde
    [(== n '())]
    [(fresh (a d)
       (== n `(,a . ,d))
       (conde
         [(== a 0) (poso d)]
         [(== a 1)])
       (olego d))]))
\end{Verbatim}
\Description[]{}
\caption{The definitions of \texttt{olego}, which unifies its argument to a ground Oleg number.}\label{fig:olego}
\end{figure}

The \texttt{olego} \texttt{defrel} is a generator for Oleg numerals.
Zero is represented as the empty list \verb|'()|.
All positive numbers are lists ending in \verb|1|.
Note the necessary \texttt{poso} required when the least significant bit of the input, \texttt{a}, is zero.
This prevents the input list, \texttt{n}, from ever ending with \texttt{0}.

Like in other logic languages, typically the meaning of a (partially instantiated) miniKanren term $p$ is taken to be the set of all its ground instances---which is to say the subset of the Herbrand Universe $\{h_{i}~\vert~h_{i} = p\theta_{i}\}$.
Here and in the following, we will for convenience represent partially instantiated terms in their reified form.
In certain cases like declarative arithmetic it makes sense to consider a partially instantiated term $p$ as describing only a subset of the ground miniKanren terms that instantiate it.
We act as though the term has some implicit constraints on it, namely, that it is only instantiated into Oleg numerals.
We say that a partially-instantiated miniKanren term $p$, is \emph{implicitly typed} as an Oleg number if the predicate \texttt{olego} succeeds on $p$.
With the symbolic representation of numbers, this means one partially instantiated term can stand for an infinite set.
For example, the term \verb|(_.0 . _.1)| represents the set of positive numbers, because it is the most specific generalization of all nonzero Oleg numbers.

Partially instantiated terms and the miniKanren arithmetic suite's symbolic representation of numbers combine to produce concise descriptions of sets of numbers.
In this context, our task of deciding if a $n$ is a multiple of three might seem straightforward---after all divisibility-by-two has a well-known canonical solution long used in hardware design, and more recently in relational programming.
A multiple of two is either zero, or a non-zero number bit-shifted to the right.
In \cref{fig:div2o}, we implement the predicate in the obvious way.
\begin{figure}[H]
\figsize{}
\begin{Verbatim}
(defrel (div2o n)
  (conde
    [(== n '())]
    [(fresh (n/2)
       (== n `(0 . ,n/2))
       (poso n/2))]))

> (run* (q) (div2o q))
'(() (0 _.0 . _.1))
\end{Verbatim}
\Description[]{}
\caption{The \texttt{div2o} relation and a query for all solutions.}\label{fig:div2o}
\end{figure}

Like in other logic languages, a partially instantiated miniKanren term $p$ is always the most general unifier for the infinite set of terms to which $p$ could be concretely instantiated.
Once again the use of \texttt{poso} enforces the implicit type constraint in the second result in the query.
The above query for all solutions to \texttt{div2o} on a fresh variable demonstrates that two solutions suffice to  describe the infinite set of divisible-by-two Oleg numbers.
With our implicit constraints, when $p =$ \verb|(0 _.0 . _.1)| it represents the set of positive even numbers, \(\{ 2n \mid n \in \mathbb{N}^+ \}\).

What makes our task more interesting is that even with these affordances, the set of Oleg numbers evenly divisible by \emph{three} cannot be precisely captured by a finite number of solutions.
\begin{lemma}\label{lemma:no-finite}
No finite set of solutions characterizes exactly the Oleg numbers divisible by three.
\end{lemma}

\begin{proof}

There are infinitely many numbers divisible by three.
The least significant bits of the sequence of numerals starting at 4 cycle as follows: $($\texttt{'(0 0)}$,$  \texttt{'(1 0)}$,$  \texttt{'(0 1)}$,$ \texttt{'(1 1)}$)$.
Divisibility by three is, definitionally, a congruence modulo 3, meaning at least one, if not two, of every four successive numbers must be divisible by three.
So for any $n > 2$ at least one multiple of three requires $n$ bits in its binary representation.
The binary representation of multiples of three grows without bound as the numbers themselves grow.

Now, suppose, toward contradiction, that $s$ is such a desired set of solutions.
For any finite set of length-instantiated partial terms, there must be a length $l$ that is the maximum length of any $p$ in the set.
Such a set necessarily excludes a multiple of three of length $l + 1$.
So for $s$ to cover all multiples of three, it must therefore contain at least one non-length-instantiated list prefix $p_i = bt_1,\ldots{},bt_k . x$, with $x$ a fresh variable in its final \texttt{cdr}.
A terminal fresh variable can compactly represent lists (numerals) of arbitrary length.
Then, for $\mathit{spfx}$ the set of numbers represented by the prefix $(\mathtt{bt_1},\ldots,\mathtt{bt_k})$, the term $p_i$ represents the set of numbers \(\{n\times2^{\floor{\log_2 \mathit{spfx}_i}} + \mathit{spfx}_i \mid \mathit{spfx}_i \in \mathit{spfx}, n \in \mathbb{N} \}\).
The prefix $(\mathtt{bt_1},\ldots,\mathtt{bt_k})$ must represent a set consisting only of multiples of three, since $x$ could represent 0 and length-instantiate the prefix.
The powers of two cycle $2^0 \equiv 1 \pmod{3}, 2^1 \equiv 2 \pmod{3}, 2^2 \equiv 1 \pmod{3} , 2^3 \equiv 2 \pmod{3}$.
This means that $2^k \pmod{3}$ will be either 1 or 2, depending whether the length of $(\mathtt{bt_1},\ldots,\mathtt{bt_k})$ is even or odd.
Given that $\mathit{spfx}_{i} \equiv 0 \pmod{3}$, $n\times2^{\floor{\log_2 \mathit{spfx}_i}} + \mathit{spfx}_i$ is equal either to $n \pmod{3}$ or $2n \pmod{3}$.
No matter which, we can find $n \in \mathbb{N}$ such that the required value is not congruent to 3, thus contradicting our initial assumption.

\end{proof}

Comparing the performance of implementations of this relation, therefore, requires considering both the time it takes to generate some prefix of the solution sequence as well as the quantity of answers covered by those solutions.

\section{Relational arithmetic solutions}\label{sec:brute-force}

In one sense \citeauthor{kiselyov2008pure} already solve our problem in the general case, $n \equiv 0 \pmod{k}$.
The ternary relation \texttt{*o} holds for numerals \texttt{a}, \texttt{b}, and \texttt{c} when $a \times b = c$.
The numeral \texttt{c} therefore represents a multiple of three when one of either the multiplier or multiplicand are \texttt{'(1 1)}.
\begin{figure}[ht]
\figsize{}
\begin{Verbatim}
(defrel (div3o/3*x n)
  (fresh (x)
    (*o '(1 1) x n)))

(defrel (div3o/x*3 n)
  (fresh (x)
    (*o x '(1 1) n)))
\end{Verbatim}
\Description[]{}
\caption{Two \texttt{div3o} implementations that use relational multiplication.}\label{fig:div3o-mults}
\end{figure}
Using their multiplication, however, still leaves a choice about whether the multiplier or multiplicand is 3.
In \cref{fig:div3o-mults}, we try it both ways: \texttt{div3o/3*x} grounds its first argument to \texttt{'(1 1)}, whereas \texttt{div3o/x*3} grounds its second.

One could perform a careful analysis of the two implementations in order to argue which will be more efficient.
We instead compare \texttt{div3o/3*x} and \texttt{div3o/x*3} with wall-clock timing in \cref{sec:performance}.
These measurements show that \texttt{div3o/3*x} and \texttt{div3o/x*3} do indeed have different performance characteristics.
Next, we introduce more \texttt{div3o} implementations that run faster than the \texttt{defrel}s based on \texttt{*o}.
\subsection{Adding still more options}\label{sec:brute-force-addition}

The \texttt{pluso} operator offers yet more possible implementations.
Although the \texttt{*o} \texttt{defrel} performs some extra work to ensure correct termination behavior in all modes, the implementations of \cref{fig:div3o-mults} use \texttt{*o} in a limited manner.
Those \texttt{*o}-derived variants may be more general than we need and less efficient than we want.
We can instead use \texttt{pluso} and express divisibility by three through repeated addition.
A number is divisible by three either when it is zero, or if it is 3 more than another number divisible by 3.
\begin{figure}[ht]
\figsize{}
\begin{Verbatim}
(defrel (div3o/x+3 n)
  (conde
    [(== n '())]
    [(fresh (x)
       (poso n)
       (pluso x '(1 1) n)
       (div3o/x+3 x))]))

(defrel (div3o/3+x n)
  (conde
    [(== n '())]
    [(fresh (x)
       (poso n)
       (pluso '(1 1) x n)
       (div3o/3+x x))]))
\end{Verbatim}
\Description[]{}
\caption{Two more \texttt{div3o} implementations that use relational addition.}
\end{figure}
We use the \texttt{pluso} relation to express the second case.
As with the brute-force multiplication \texttt{defrel}s, there are two approaches depending whether we instantiate the augend or the addend with the ground constant \texttt{'(1 1)}.

Unlike the \texttt{div3o} implementations based on \texttt{*o}, these \texttt{pluso}-based implementations do not explicitly find the factor being multiplied with 3.
However, the other factor appears implicitly as the number of repeated additions performed;
In \cref{sec:exploiting-the-structure-of-oleg-numbers}, we introduce two more \texttt{div3o} implementations that do not, directly or indirectly, depend on the other factor.

\section{Exploiting the structure of Oleg numbers}\label{sec:exploiting-the-structure-of-oleg-numbers}
A base-ten number is divisible by three if and only if the sum of its digits is divisible by three.
There is a similarly well-known technique for testing the divisibility by three of a binary number: a binary number is divisible by three if and only if the sum of its bits at even indices minus the sum of its bits at odd indices is divisible by three.
Mathematically, this is expressible as:
\[ \sum_{\text{odd } i} x_i - \sum_{\text{even } i} x_i \equiv 0 \bmod 3 \]

\noindent where $x_{i}$ indicates the element at the $i$th position of the numeral \verb|x|.
This congruence mod 3 implies we need consider only three states while summing the bits at even and odd indices, which in turn implies we can decide the problem via a simple DFA\footnote{We omit the derivation of this DFA; an explanation is available at \url{https://cs.stackexchange.com/a/7889}.} like that of \cref{fig:minimized-dfa}.
At any point in the DFA, the Oleg numeral is 0, 1, or 2 state transitions from the accepting state, represented by \texttt{'q1}.
While in \texttt{'q1} one \texttt{conde} branch will unify the \texttt{cdr} with the empty list, terminating the recursion and generating a solution.
The other branch continues recursion by grounding the \texttt{cadr} to 0 or 1. 
If the \texttt{cadr} is 0, the congruence remains unaffected. 
If the \texttt{cadr} is 1, the congruence is disrupted, causing the DFA to transition out of the accepting state. 
It will return to the accepting state once the congruence mod 3 is restored to 0 by encountering another 1 as represented by the transition from \texttt{'q2} to \texttt{'q1}.
\texttt{'q2} and \texttt{'q3} represent states where the sum of bits mod 3 is not 0.
The DFA in these states is one transition away from the accepting state, needing either one evenly or oddly indexed bit to restore the congruence mod 3 to 0.
\texttt{'q2} transitions to the accepting state upon receiving a 1, while \texttt{'q3} transitions to \texttt{'q2} upon receiving a 0. 
If \texttt{'q2} or \texttt{'q3} receives a 0, they oscillate until \texttt{'q2} receives a 1 and transitions to the accepting state.
Finally, when \texttt{'q3} generates a 1, the DFA will remain at \texttt{'q3}. 

\begin{figure}[ht]
\figsize{}
\centering
\begin{tikzpicture}[scale=0.2]
\tikzstyle{every node}+=[inner sep=0pt]
\draw [black] (19.6,-31.2) circle (3);
\draw (19.6,-31.2) node {$q_1$};
\draw [black] (19.6,-31.2) circle (2.4);
\draw [black] (34.7,-31.2) circle (3);
\draw (34.7,-31.2) node {$q_2$};
\draw [black] (49.2,-31.2) circle (3);
\draw (49.2,-31.2) node {$q_3$};
\draw [black] (13.2,-31.2) -- (16.6,-31.2);
\fill [black] (16.6,-31.2) -- (15.8,-30.7) -- (15.8,-31.7);
\draw [black] (22.169,-29.664) arc (113.88717:66.11283:12.302);
\fill [black] (32.13,-29.66) -- (31.6,-28.88) -- (31.2,-29.8);
\draw (27.15,-28.11) node [above] {$1$};
\draw [black] (32.205,-32.85) arc (-63.97796:-116.02204:11.522);
\fill [black] (22.1,-32.85) -- (22.59,-33.65) -- (23.03,-32.75);
\draw (27.15,-34.52) node [below] {$1$};
\draw [black] (18.277,-28.52) arc (234:-54:2.25);
\draw (19.6,-23.95) node [above] {$0$};
\fill [black] (20.92,-28.52) -- (21.8,-28.17) -- (20.99,-27.58);
\draw [black] (47.877,-28.52) arc (234:-54:2.25);
\draw (49.2,-23.95) node [above] {$1$};
\fill [black] (50.52,-28.52) -- (51.4,-28.17) -- (50.59,-27.58);
\draw [black] (37.412,-29.931) arc (108.93363:71.06637:13.986);
\fill [black] (46.49,-29.93) -- (45.89,-29.2) -- (45.57,-30.14);
\draw (41.95,-28.67) node [above] {$0$};
\draw [black] (46.755,-32.92) arc (-62.99516:-117.00484:10.581);
\fill [black] (37.15,-32.92) -- (37.63,-33.73) -- (38.09,-32.84);
\draw (41.95,-34.57) node [below] {$0$};
\end{tikzpicture}
\Description[]{A minimized DFA for 3-divisible Oleg number.}
\caption{A minimized DFA that accepts bitstrings divisible by 3}\label{fig:minimized-dfa}
\end{figure}
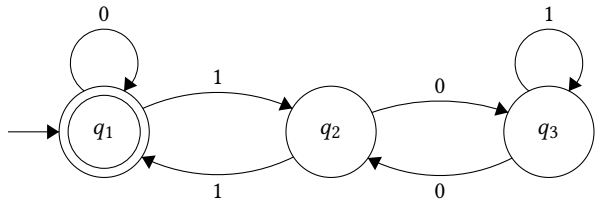

Converting a DFA into a miniKanren relation is straightforward.
We proceed recursively through the list while maintaining a state (here
represented as a symbol), and only accept an empty list if the current
state is an accepting state.

We present the result of this translation in \cref{fig:dfao}. The
translation from DFA to miniKanren clauses is mostly direct with one
exception: encountering a \texttt{0} requires the tail of the list to be
positive (a pair instead of the empty list), since Oleg numerals
cannot end in \texttt{0}. A \texttt{div3o} relation can utilize \texttt{dfao}
and passing in the starting state \texttt{'q1}.

\begin{figure}[ht]
\figsize{}
\begin{Verbatim}
(defrel (dfao l state)
  (conde
    [(== l '()) (== state 'q1)]
    [(fresh (a d next-state)
       (== l `(,a . ,d))
       (conde
         [(== a 0) (poso d)
          (conde
            [(== state 'q1) (== next-state 'q1)]
            [(== state 'q2) (== next-state 'q3)]
            [(== state 'q3) (== next-state 'q2)])]
         [(== a 1)
          (conde
            [(== state 'q1) (== next-state 'q2)]
            [(== state 'q2) (== next-state 'q1)]
            [(== state 'q3) (== next-state 'q3)])])
        (dfao d next-state))]))

(defrel (div3o/dfa n)
  (dfao n 'q1))
\end{Verbatim}
\Description[]{}
  \caption{The miniKanren relation derived from the DFA of \cref{fig:minimized-dfa} and a wrapper for calling it with the correct initial state.}\label{fig:dfao}
\end{figure}

\subsection{A relation that leaves bits fresh}\label{sec:a-relation-that-leaves-bits-fresh}
The previous implementation, \texttt{div3o/dfa}, uses a well-known recursion scheme over lists: the empty list base case, and the non-empty list recursive case where the list has a first element, \texttt{a}, and a tail, \texttt{d}.
Here we introduce another \texttt{div3o} implementation which we develop from a different recursion scheme.
We instead consider the cases of an empty list, a singleton list containing one element, and a list containing two or more elements: a first element \texttt{a}, a second element \texttt{ad}, and a tail \texttt{dd}.
\footnote{The names \texttt{a}, \texttt{d}, \texttt{ad}, and \texttt{dd} are used by convention. They are derived from the Lisp functions \texttt{car}, \texttt{cdr}, \texttt{cadr}, and \texttt{cddr}, respectively.}
We consider \texttt{a} to be at index 0 and \texttt{ad} to be at index 1, in accordance with \texttt{list-ref}.

This alternate recursion scheme deals with exactly one even-indexed element and one odd-indexed element.
So if the parent call deals with indices 0 and 1, the child call deals with indices 2 and 3, the grandchild call, 4 and 5, and so on.
Therefore any \texttt{a} is at an even index, and \texttt{ad} is at an odd index.

A sketch of our implementation is as follows: Given a list representing an Oleg number \texttt{n}, maintain a difference modulo 3 of the even-indexed bits and the odd-indexed bits, named \texttt{diff}. This argument, \texttt{diff}, is only ever 0, 1, or 2, because it represents a difference modulo 3.
It succeeds only if the difference is 0.
We call this implementation \texttt{div3o/even-odd}. We use a helper relation, \texttt{even-odd-helpero}, to carry the \texttt{diff} along with the Oleg number \texttt{n}.

We make a slight modification to this two-at-a-time recursion scheme.
Like in the definitions of \texttt{olego} and \texttt{dfao}, we strategically add a \texttt{poso} into the recursive case to ensure that our nonzero results end in a 1.
This requires us to consider the Oleg numbers up to 3 separately, and limit the recursive call to operate only on Oleg numbers 4 or greater.

\begin{figure}[ht]
\figsize{}
\begin{Verbatim}
(defrel (even-odd-helpero n diff)
  (conde
    [(== n '()) (== diff 0)]
    [(== n '(1)) (== diff 2)]
    [(== n '(0 1)) (== diff 1)]
    [(== n '(1 1)) (== diff 0)]
    [(fresh (a ad dd new-diff)
       (== n `(,a ,ad . ,dd))
       (poso dd)
       (conde
         [(== a ad) (== diff new-diff)]
         [(== `(,a ,ad) '(0 1)) (+1mod3o new-diff diff)]
         [(== `(,a ,ad) '(1 0)) (+1mod3o diff new-diff)])
       (even-odd-helpero dd new-diff))]))

(defrel (div3o/even-odd n)
  (even-odd-helpero n 0))
\end{Verbatim}
\Description[]{}
\caption{The \texttt{even-odd} implementation of \texttt{div3o}, and its helper relation \texttt{even-odd-helper}.}
\end{figure}

The \texttt{defrel} \texttt{div3o/even-odd} calls \texttt{even-odd-helpero} with a constant 0 passed to \texttt{diff}, asserting that the difference in even and odd bits should be zero if and only if \texttt{n} is divisible by 3.
When the first two elements of the list, \texttt{a} and \texttt{ad}, are equal, there is no change in the difference; the adjacent even and odd bits cancel each other out.
There are two ways that \texttt{a} and \texttt{ad} could be equal, but we do not enumerate them. Instead, we leave the bits fresh. This is useful for when we compare our implementations in \cref{sec:performance}.

However, we enumerate both cases which cause \texttt{a} and \texttt{ad} to be unequal.
First, if \texttt{a} is 1 and \texttt{ad} is 0, the new difference is one more than the old difference. Second, if \texttt{a} is 0 and \texttt{ad} is 1, the new difference is one \emph{less} than the old difference.

To increment and decrement \texttt{diff}, we define a two-place successor-modulo-3 relation \texttt{+1mod3o}. Incrementing is straightforward. Decrementing can be achieved by swapping the order of arguments.
Like \texttt{poso} and \texttt{div2o} from \cref{sec:minikanren-arithmetic}, this relation is non-recursive.
Therefore, \texttt{+1mod3o} is placed before any recursive relations.\footnote{Chapter 4 of \citetitle{friedman2018reasoned} explains why non-recursive relations should be placed before recursive relations.}
\begin{figure}[ht]
\figsize{}
\begin{Verbatim}
(defrel (+1mod3o n n+1)
  (conde
    [(== n 0) (== n+1 1)]
    [(== n 1) (== n+1 2)]
    [(== n 2) (== n+1 0)]))
\end{Verbatim}
\Description[]{}
\caption{A non-recursive relation asserting the successor of a number \texttt{n} mod 3.}
\end{figure}

The results of querying \texttt{div3o/even-odd} on a fresh variable is different than the other \texttt{div3o} implementations. Some of the bits remain as free logic variables, specifically when two adjacent bits are equal.
Because of this, we need to interpret the results differently.
Here, each solution does not represent a multiple of three, but rather a \emph{family} of multiples of three.
For example, the sixth solution \verb|(_.0 _.0 _.1 _.1 1 1)| expresses all multiples of three \(x + 2x + 4y + 8y + 16 + 32 = 3(x + 4y + 16)\) for any assignment of bits \(x, y \in \{0, 1\}\).

\begin{figure}[ht]
\figsize{}
\begin{Verbatim}[]
> (run 10 (q) (div3o/even-odd q))
'(()
  (1 1)
  (_.0 _.0 1 1)
  (0 1 1)
  (1 0 0 1)
  (_.0 _.0 _.1 _.1 1 1)
  (_.0 _.0 0 1 1)
  (_.0 _.0 1 0 0 1)
  (0 1 _.0 _.0 1)
  (1 0 _.0 _.0 0 1))
\end{Verbatim}
\Description[]{}
\caption{The first 10 results of running \texttt{div3o/even-odd} on a fresh variable \texttt{q}.}
\end{figure}

The first 10 answers of the above \texttt{div3o/even-odd} query with a fresh variable are interesting, because they leave many bits of the Oleg number fresh. As noted earlier, terms are implicitly constrained to represent Oleg numbers, and so the list elements are likewise constrained to $\{\mathtt{0}, \mathtt{1}\}$.
We discuss the notion of performance of the preceding definitions in the next section.

\section{Performance}\label{sec:performance}

Benchmarking these implementations presents an uncommon situation.
We can perform the same kind of query against each of our implementations.
For each of our implementations, a query with a fresh variable has infinitely many solutions, where each solution represents some \emph{finite} quantity of answers---and not generally same quantity.
It's uncommon to have a collection similar queries across different implementations of a relation, each of which produces a solution set with this property.
This raises some unique questions for how to benchmark implementations.

\subsection{How to compare implementations}\label{sec:how-to-compare-relations}

We first consider as measure the number of answers each solution represents.
The number of answers each solution represents will differ from solution to solution from implementation to implementation.
Thus, it makes sense to also track the number of potential ground solutions produced from each \texttt{defrel} implementation as a function the first \(i\) results in its solution sequence.

The benchmarks below operate on a higher-order argument, a reference to the specific \texttt{div3o} implementation being tested.
Given a \texttt{defrel} that succeeds on any Oleg multiple of three and fails on any Oleg non-multiple-of-three, and given natural numbers \(n\) and \(i\), below are some questions you can ask. If you ask these questions with a larger and larger values of \(n\) and \(i\), you can get closer to a theory of asymptotic analysis on relational programs.

\subsubsection{Speed}
The first test is straightforward. How long does it take the \texttt{div3o} implementation to succeed or fail on a ground Oleg number $n$?
This test cannot be performed on \texttt{defrel}s that diverge on failure.
Thankfully, this does not apply to any of our six \texttt{div3o} implementations.
\begin{figure}[ht]
\figsize{}
\begin{Verbatim}
(define (time-of-run* div3o n)
  (cpu-time
    (lambda ()
      (run* (q) (div3o (build-num n))))))
\end{Verbatim}
\Description[]{}
\caption{A host-level function that times how long a \texttt{div3o} implementation takes to halt on the Oleg form of the Racket number \texttt{n}.}
\end{figure}

The helper function \texttt{cpu-time} is defined in \cref{sec:appendix}. It evaluates its argument thunk, discards the result, and returns the elapsed CPU time in milliseconds.
We also make use of the \texttt{build-num}\footnote{\Citet{friedman2018reasoned} give a more comprehensive explanation of \texttt{build-num}.} function to convert a Racket number into its corresponding Oleg numeral.

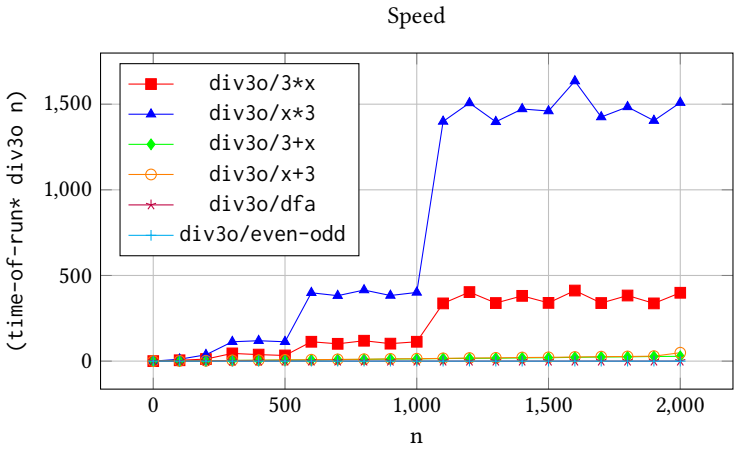
\begin{figure}[ht]
  \figsize{}
  \centering
  \begin{tikzpicture}
    \begin{axis}[
        title={Speed},
        xlabel={n},
        ylabel={\texttt{(time-of-run* div3o n)}},
        legend pos=north west,
        grid=both,
        width=10cm,
        height=6cm,
    ]

    \addplot[color=red, mark=square*]
    coordinates {
      (0,0)
      (100,5)
      (200,12)
      (300,45)
      (400,38)
      (500,33)
      (600,113)
      (700,101)
      (800,119)
      (900,102)
      (1000,113)
      (1100,337)
      (1200,403)
      (1300,339)
      (1400,381)
      (1500,340)
      (1600,412)
      (1700,339)
      (1800,383)
      (1900,337)
      (2000,399)
    };
    \addlegendentry{\texttt{div3o/3*x}}

    \addplot[color=blue, mark=triangle*]
    coordinates {
      (0,0)
      (100,11)
      (200,37)
      (300,113)
      (400,119)
      (500,113)
      (600,399)
      (700,382)
      (800,415)
      (900,383)
      (1000,401)
      (1100,1399)
      (1200,1507)
      (1300,1397)
      (1400,1472)
      (1500,1460)
      (1600,1635)
      (1700,1425)
      (1800,1484)
      (1900,1404)
      (2000,1509)
    };
    \addlegendentry{\texttt{div3o/x*3}}

    \addplot[color=green, mark=diamond*]
    coordinates {
      (0,0)
      (100,1)
      (200,2)
      (300,4)
      (400,5)
      (500,6)
      (600,7)
      (700,9)
      (800,10)
      (900,12)
      (1000,13)
      (1100,15)
      (1200,16)
      (1300,17)
      (1400,19)
      (1500,20)
      (1600,22)
      (1700,24)
      (1800,25)
      (1900,27)
      (2000,28)
    };
    \addlegendentry{\texttt{div3o/3+x}}

    \addplot[color=orange, mark=o]
    coordinates {
      (0,0)
      (100,1)
      (200,2)
      (300,4)
      (400,5)
      (500,7)
      (600,8)
      (700,10)
      (800,11)
      (900,13)
      (1000,14)
      (1100,16)
      (1200,18)
      (1300,19)
      (1400,21)
      (1500,22)
      (1600,24)
      (1700,26)
      (1800,27)
      (1900,29)
      (2000,49)
    };
    \addlegendentry{\texttt{div3o/x+3}}

    \addplot[color=purple, mark=star]
    coordinates {
      (0,0)
      (100,0)
      (200,0)
      (300,0)
      (400,0)
      (500,0)
      (600,0)
      (700,0)
      (800,0)
      (900,0)
      (1000,0)
      (1100,0)
      (1200,0)
      (1300,0)
      (1400,0)
      (1500,0)
      (1600,0)
      (1700,0)
      (1800,0)
      (1900,0)
      (2000,0)
    };
    \addlegendentry{\texttt{div3o/dfa}}

    \addplot[color=cyan, mark=+]
    coordinates {
      (0,0)
      (100,0)
      (200,0)
      (300,0)
      (400,0)
      (500,0)
      (600,0)
      (700,0)
      (800,0)
      (900,0)
      (1000,0)
      (1100,0)
      (1200,0)
      (1300,0)
      (1400,0)
      (1500,0)
      (1600,0)
      (1700,0)
      (1800,0)
      (1900,0)
      (2000,0)
    };
    \addlegendentry{\texttt{div3o/even-odd}}
  
    \end{axis}
  \end{tikzpicture}
\Description[]{}
\caption{How long does it take (average of three runs, in milliseconds) for each \texttt{defrel} to decide if a given $n \equiv 0 \pmod{3}$? This graphs indicates that \texttt{div3o/x*3} and \texttt{div3o/3*x} perform much worse than the other implementations, especially after $n = 255$, $n = 511$, and $n = 1023$. The increased bitlength of $n$ at each power of two explains these inflection points.}\label{fig:all-speed}
\end{figure}

We drop \texttt{div3o/3*x} and \texttt{div3o/x*3} from consideration and perform the same benchmark on the remaining implementations for much higher values of $n$.

\begin{figure}[ht]
  \figsize{}
  \centering
  \begin{tikzpicture}
    \begin{axis}[
        title={Speed},
        xlabel={n},
        ylabel={\texttt{(time-of-run* div3o n)}},
        legend pos=north west,
        grid=both,
        width=10cm,
        height=6cm,
        tick scale binop=\times,
        scaled ticks=false,
        yticklabel style={
            /pgf/number format/.cd,
                fixed,
                fixed zerofill,
                precision=0
        },
        xticklabel style={
            /pgf/number format/.cd,
                fixed,
                fixed zerofill,
                precision=0
        },
    ]

    \addplot[color=green, mark=diamond*]
    coordinates {
      (0,0)
      (10000,179)
      (20000,364)
      (30000,590)
      (40000,826)
      (50000,1094)
      (60000,1325)
      (70000,1600)
      (80000,1833)
      (90000,2104)
      (100000,2376)
      (110000,2647)
      (120000,2933)
      (130000,3233)
      (140000,3518)
      (150000,3839)
      (160000,4060)
      (170000,4398)
      (180000,4698)
      (190000,5035)
      (200000,5311)
    };
    \addlegendentry{\texttt{div3o/3+x}}

    \addplot[color=orange, mark=o]
    coordinates {
      (0,0)
      (10000,175)
      (20000,390)
      (30000,628)
      (40000,879)
      (50000,1170)
      (60000,1411)
      (70000,1695)
      (80000,1971)
      (90000,2266)
      (100000,2569)
      (110000,2854)
      (120000,3166)
      (130000,3451)
      (140000,3772)
      (150000,4078)
      (160000,4376)
      (170000,4710)
      (180000,5025)
      (190000,5346)
      (200000,5727)
    };
    \addlegendentry{\texttt{div3o/x+3}}

    \addplot[color=purple, mark=star]
    coordinates {
      (0,0)
      (10000,0)
      (20000,0)
      (30000,0)
      (40000,0)
      (50000,0)
      (60000,0)
      (70000,0)
      (80000,0)
      (90000,0)
      (100000,0)
      (110000,0)
      (120000,0)
      (130000,0)
      (140000,0)
      (150000,0)
      (160000,0)
      (170000,0)
      (180000,0)
      (190000,0)
      (200000,0)
    };
    \addlegendentry{\texttt{div3o/dfa}}

    \addplot[color=cyan, mark=+]
    coordinates {
      (0,0)
      (10000,0)
      (20000,0)
      (30000,0)
      (40000,0)
      (50000,0)
      (60000,0)
      (70000,0)
      (80000,0)
      (90000,0)
      (100000,0)
      (110000,0)
      (120000,0)
      (130000,0)
      (140000,0)
      (150000,0)
      (160000,0)
      (170000,0)
      (180000,0)
      (190000,0)
      (200000,0)
    };
    \addlegendentry{\texttt{div3o/even-odd}}
  
    \end{axis}
  \end{tikzpicture}
\Description[]{}
  \caption{How long does it take (average of three runs, in miliseconds) for each \emph{fast} \texttt{defrel} to decide if a given $n \equiv 0 \pmod{3}$?
}\label{fig:fast-speed}
\end{figure}
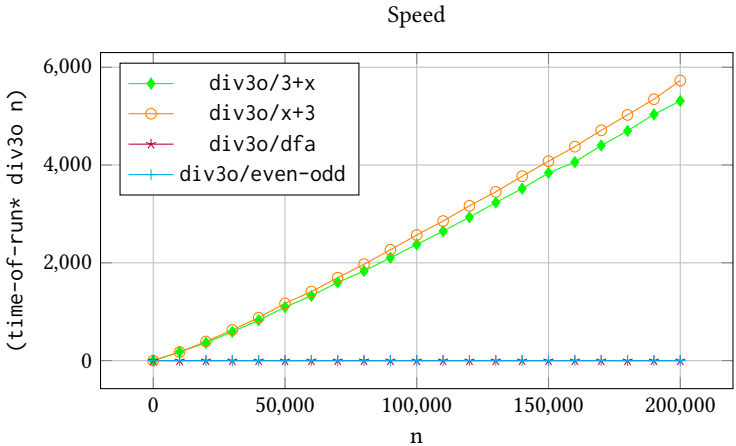

At this scale there is a noticeable speed difference between the addition-based implementations and the DFA-based implementations.
The results of this test point us towards either using \texttt{div3o/dfa} or \texttt{div3o/even-odd} as the canonical implementation.
Among those two, it seems better to use \texttt{div3o/dfa} because it consumes its input in a more standard recursive style. It is a simpler implementation.
Still, we argue that there is something preferable about the \texttt{div3o/even-odd} implementation. It requires a different way of quantifying the \texttt{defrel}'s performance.

\subsubsection{Reach}
There is another heuristic that we argue is useful: how many ground Oleg numbers are subsumed by the first \(i\) members of the solution sequence?
We call this a relation's \emph{reach}\footnote{Our \texttt{reach} implementation could be generalized to possibly infinite unary Oleg relations, like \texttt{div2o}, but not concisely.}.
A relation's reach is evidence of its non-ground solutions. If a relation has a superlinear reach, then its solution sequence necessarily contains logic variables.
Relations with greater reach may perform better in conjunctions. Ideally, the first conjunct succeeds a finite number of times, so that the second relation must only consider a constant number of substitutions to extend or fail on.
But there are cases where the first conjunct does not succeed finitely.
For example, any \texttt{div3o} implementation run on fresh variables must succeed infinitely often, as proved in \cref{lemma:no-finite}.
In a general sense, a relation's reach should approach the size of its solution space. As in the working example of finding numbers divisible by three, a relation that leaves variables fresh will approach the infinite solution set faster.
In problems where the solution space is finite, such as a sudoku puzzle, an ideal relation's reach would be the number of possible solutions.
\begin{figure}[ht]
\figsize{}
\begin{Verbatim}
(define (reach div3o i)
  (set-count
    (append-map all-solutions (run i (q) (div3o q)))))
\end{Verbatim}
\Description[]{}
\caption{The \texttt{reach} function counts how many multiples of 3 the first \(i\) partial terms describe in a \texttt{div3o} relation.}\label{fig:partials-described}
\end{figure}
\begin{figure}[ht]
\figsize{}
\centering
\begin{tikzpicture}
  \begin{axis}[
      title={Reach},
      xlabel={i},
      ylabel={\texttt{(reach div3o i)}},
      legend pos=north west,
      grid=both,
      width=10cm,
      height=6cm,
      tick scale binop=\times,
      scaled ticks=false,
      yticklabel style={
          /pgf/number format/.cd,
              fixed,
              fixed zerofill,
              precision=0
      },
      xticklabel style={
          /pgf/number format/.cd,
              fixed,
              fixed zerofill,
              precision=0
      },
  ]
  \addplot[color=purple, mark=star]
      coordinates {
        (0,0)
        (1000,1000)
        (2000,2000)
        (3000,3000)
        (4000,4000)
        (5000,5000)
        (6000,6000)
        (7000,7000)
        (8000,8000)
        (9000,9000)
        (10000,10000)
        (11000,11000)
        (12000,12000)
        (13000,13000)
        (14000,14000)
        (15000,15000)
        (16000,16000)
        (17000,17000)
        (18000,18000)
        (19000,19000)
        (20000,20000)
      };
  \addlegendentry{\texttt{div3o/dfa}}

  \addplot[
      color=cyan,
      mark=+
      ]
      coordinates {
        (0,0)
        (1000,14350)
        (2000,39216)
        (3000,81918)
        (4000,107158)
        (5000,153206)
        (6000,223846)
        (7000,266598)
        (8000,292778)
        (9000,353566)
        (10000,418614)
        (11000,588222)
        (12000,611494)
        (13000,649022)
        (14000,728358)
        (15000,754342)
        (16000,799814)
        (17000,880602)
        (18000,965974)
        (19000,1018790)
        (20000,1143606)
      };
  \addlegendentry{\texttt{div3o/even-odd}}

  \end{axis}
\end{tikzpicture}
\Description[]{}
\caption{Since \texttt{div3o/dfa} grounds its argument, the first 20,000 solutions of \texttt{div3o/dfa} implementation subsumes 20,000 multiples of 3. This is true for all of our prior \texttt{div3o} implementations, because they leave no bits fresh. On the other hand, the first 20,000 solutions of \texttt{div3o/even-odd} subsumes 1,143,606 multiples of 3.}\label{fig:reach}
\end{figure}
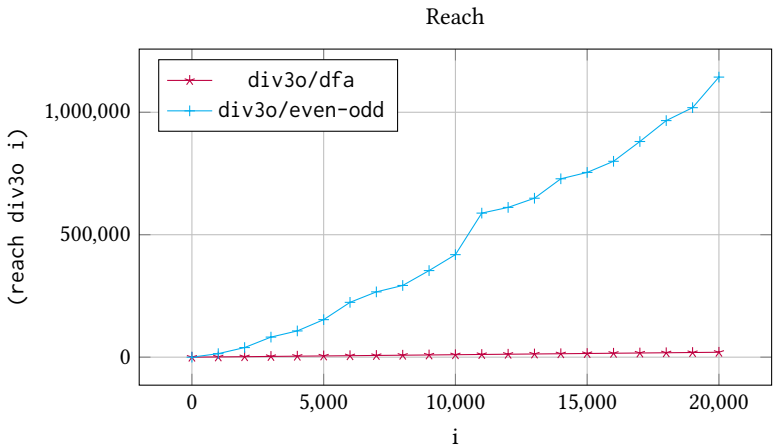

A \texttt{defrel} like \texttt{div3o/x*3} has linear reach, because \texttt{div3o/x*3} always conjures ground results. Since \texttt{div3o/even-odd} leaves some bits of its solution set fresh, it has superlinear reach, as seen in \cref{fig:reach}.
\section{Related Work}\label{sec:related}

There is of course a wealth of research in benchmarking for logic programming.
This benchmarking typically uses a fixed set of illustrative programs~\cite{haygood1989prolog} as the \emph{data} for comparison.
This could be in the service of comparing the performance of the code an particular implementation generates for different targets, comparing implementations of the same language~\cite{fernandezleiva2000comparative,demoen2001odd}, or even for drawing a comparison between different languages~\cite{warren1977prolog,somogyi1995mercury} or paradigms~\cite{vanroy1990can}.

Comparing the performance of two implementations of the same or closely related relations comes up more often in the context of debugging and profiling.
Previous work on profiling Prolog programs can be found in \citet{debray1988profiling}, which focuses on adding timing facilities within Prolog itself.
We have the benefit of working with an embedded logic language, so we instead opt to treat a \texttt{run} of our embedded relational programs as a black box and use timing facilities from the host language.
Other approaches to profiling deal with a first order structure of the logic program by tracking the call graph~\cite{spivey2004fast,mera2011profiling} or the program points~\cite{minikanrenprofiler}.
\Citet{minikanrenprofiler} describe recent work toward writing performant relations in miniKanren, where their focus is failing early to cut out parts of the search.

Such profiling tools could also show the differences in performance and may offer additional insight into the behavior.
We implement vastly different algorithms for solving the problem, so it's not clear how much guidance profiling would offer to a programmer using the wrong algorithm.
Our process of converting finite automata to miniKanren code follows \citeposs{chomskyhierarchy}.

Far and away the closest work is \citet{kiselyov2008pure}.
We rely heavily on their arithmetic suite and their representation of numerals, and their predicates for implicit typing.
Their focus is on arriving at arithmetic relations with good termination properties; they compare implementations that differ in termination behavior, but they are not particularly concerned with their performance beyond termination.
Unlike their work, our focus is defining and measuring performance---both of arithmetic-based implementations and otherwise---for a particular problem.

There is of course a long history of ternary computer arithmetic.
A relational ternary arithmetic\footnote{For instance, like that described in \url{https://homepage.cs.uiowa.edu/~dwjones/ternary/}} suite for miniKanren would trivialize \texttt{div3o} in same way as \texttt{div2o} on binary numbers.

\section{Conclusion and Future Work}\label{sec:conclusion}

We investigated several definitions of the relation \enquote{evenly divisible by three} concretely implemented using miniKanren.
We proved that no \texttt{div3o} implementation could finitely represent all the solutions---every implementation must enumerate an infinite solution sequence.
Our present question about measuring performance arose because repeated evaluation of a goal has an infinite solution set, but the solutions describe varying, though always finite, numbers of answers.

Future work includes finding a \texttt{div3o} implementation with even better performance, reach, or both.
There is also an opportunity to generalize the \texttt{reach} function to count the ground solution sets of any relation, not just the unary relations whose argument is implicitly typed as an Oleg number.

\begin{acks}
We thank the anonymous reviewers for their suggestions and improvements. We also thank Will Byrd for his recommendation to benchmark these implementations using \texttt{faster-miniKanren}. The research of Jason Hemann and Brysen Pfingsten has been partially supported by \grantsponsor{nsf}{NSF}{https://www.nsf.gov} grant \grantnum{nsf}{CCF-2348408}.
\end{acks}

\printbibliography{}

\pagebreak

\section{Appendix}\label{sec:appendix}

\begin{figure}[ht]
  \figsize{}
  \begin{Verbatim}
(define (ground i)
  (or (eqv? i 0) (eqv? i 1)))

(define (find-first-reified num)
  (let ([reifieds (filter-not ground num)])
    (and (pair? reifieds) (car reifieds))))

(define (((n-or-input n) vs) i)
  (if (eqv? i vs) n i))

(define one-or-input (n-or-input 1))
(define zero-or-input (n-or-input 0))

(define (helper queue)
  (cond
    [(find-first-reified (car queue))
     => (lambda (fr)
          (helper
            (append
              (cdr queue)
              (list
                (map (zero-or-input fr) (car queue))
                (map (one-or-input fr) (car queue))))))]
    [else queue]))

(define (all-solutions b)
  (cond
    [(null? b) '(())]
    [(ground (last b)) (helper (list b))]
    [else (helper (list (map (one-or-input (last b)) b)))]))
  \end{Verbatim}
  \Description[]{}
  \caption{A set of Racket functions to return the list of unique Oleg numbers subsumed by the term \texttt{b}, which may contain reified variables like \texttt{\_.0}.}
\end{figure}
\end{document}

\begin{figure}[ht]
  \figsize{}
  \begin{Verbatim}
  (define (cpu-time proc)
    (let-values ([(result cpu real garbage-collection) (time-apply proc '())])
      cpu))
  \end{Verbatim}
  \Description[]{}
  \caption{A Racket function which runs \texttt{proc} with no arguments, discards the result, and instead returns the CPU time elapsed, in milliseconds.}
\end{figure}

\begin{figure}[ht]
  \figsize{}
  \begin{Verbatim}
(define fast-div3o-impls
  (list
    div3o/3+x
    div3o/x+3
    div3o/dfa
    div3o/even-odd))

(define div3o-impls
  (list*
    div3o/3*x
    div3o/x*3
    fast-div3o-impls))
  \end{Verbatim}
  \Description[]{}
  \caption{A list of all the \texttt{div3o} implementations, preceded by a list containing just the faster ones.}
\end{figure}

\begin{Verbatim}
(require math/statistics)

(define ((multi-run-avg i) rel n)
  (round
    (mean
      (for/list ((i (range 0 i)))
        (time-of-run* rel n)))))

(define three-run-avg (multi-run-avg 3))

(define (plot-time/each impls input-range)
  (for-each
    (lambda (div3o)
      (displayln (object-name div3o))
      (for-each
        (lambda (n)
          (printf "~a,~a\n" n (three-run-avg div3o n)))
        input-range))
    impls))

(define (plot-time-run*)
  (plot-time/each div3o-impls (range 0 2100 100)))

(define (plot-time-run*-fast)
  (plot-time/each fast-div3o-impls (range 0 210000 10000)))
\end{Verbatim}